\documentclass[twocolumn,superscriptaddress,aps,showpacs,amsmath,amstex,amssymb,citeautoscript,floatfix,pra,preprintnumbers]{revtex4-2}
\pdfoutput=1
\usepackage[english]{babel}

\usepackage{letltxmacro}
\usepackage{latexsym}
\usepackage{booktabs}
\usepackage{siunitx}
\usepackage{multirow}
\usepackage{algpseudocode}
\usepackage{bbm}

\LetLtxMacro{\ORIGselectlanguage}{\selectlanguage}
\makeatletter
\DeclareRobustCommand{\selectlanguage}[1]{%
  \@ifundefined{alias@\string#1}
    {\ORIGselectlanguage{#1}}
    {\begingroup\edef\x{\endgroup
       \noexpand\ORIGselectlanguage{\@nameuse{alias@#1}}}\x}%
}
\newcommand{\definelanguagealias}[2]{%
  \@namedef{alias@#1}{#2}%
}
\makeatother

\definelanguagealias{en}{english}
\definelanguagealias{English}{english}
\usepackage{graphicx}
\usepackage{amsmath}
\usepackage{amsfonts}
\usepackage{amssymb}
\usepackage{amsthm}
\usepackage{cancel}
\usepackage{mathtools}
\usepackage{bm}
\usepackage{color}
\usepackage[percent]{overpic}
\usepackage{soul}
\usepackage{wasysym}
\usepackage{dsfont}
\usepackage{float}
\usepackage{physics}
\usepackage{hyperref}
\usepackage{comment}
\usepackage{enumitem}
\usepackage{textgreek}
\usepackage{tcolorbox}
\usepackage[dvipsnames]{xcolor}

\newcounter{alg}

\hypersetup{
  colorlinks   = true,
  urlcolor     = blue,
  linkcolor    = blue,
  citecolor    = red
}
\usepackage[mathscr]{euscript}

\usepackage{varwidth}
\usepackage[lined,boxed,ruled,norelsize,linesnumbered]{algorithm2e}
\usepackage{verbatim}
\usepackage[normalem]{ulem}
\usepackage{cleveref}
\usepackage{times}

\newtheorem{theorem}{Theorem}

\newtheorem{lemma}{Lemma}

\newtheorem{proposition}{Proposition}

\newtheorem*{lemma*}{Lemma}
\newtheorem{corollary}{Corollary}
\newtheorem*{theorem*}{Theorem}

\theoremstyle{plain}

\theoremstyle{plain}

\theoremstyle{plain}
\newtheorem*{lem*}{\protect\lemmaname}
\theoremstyle{plain}
\newtheorem*{thm*}{\protect\theoremname}
\theoremstyle{plain}

\theoremstyle{plain}

\renewcommand{\thealg}{\arabic{alg}}
\newtcolorbox[use counter=alg,
              crefname={algorithm}{algorithms},
              Crefname={Algorithm}{Algorithms}]
{alg}[2][]{%
  floatplacement=#1,
  float,
  colback=cyan!5!white,
  colframe=cyan!50!black,
  colbacktitle=cyan!85!black,
  fonttitle=\bfseries,
  title=Algorithm~\thealg: #2
}

\newcommand{\eqs}[1]{\begin{equation}\begin{split}#1\end{split}\end{equation}}

\newcommand{\corrauthors}{%
\href{mailto:hongyehu@iphy.ac.cn}{hongyehu@iphy.ac.cn}, \href{mailto:txiang@iphy.ac.cn}{txiang@iphy.ac.cn},
\href{mailto:chenzhao@eitech.edu.cn}{chenzhao@eitech.edu.cn}
}
\begin{document}

\title{
Fast and accurate neural triage decoding of ultra-high-rate quantum LDPC codes
}

\author{Hong-Ye Hu}
\thanks{\corrauthors}
\affiliation{Institute of Physics, Chinese Academy of Sciences, Beijing 100910, China}

\author{Zhenlai Wang}
\affiliation{Institute of Physics, Chinese Academy of Sciences, Beijing 100910, China}

\author{Binzhao Luo}
\affiliation{Institute of Physics, Chinese Academy of Sciences, Beijing 100910, China}

\author{Tao Xiang}
\thanks{\corrauthors}
\affiliation{Institute of Physics, Chinese Academy of Sciences, Beijing 100910, China}
\affiliation{School of Physical Sciences, University of Chinese Academy of Sciences, Beijing 100049, China}

\author{Chen Zhao}
\thanks{\corrauthors}
\affiliation{School of Mathematical Sciences, Eastern Institute of Technology, Ningbo, Zhejiang 315200, China}

\begin{abstract}
Quantum error correction is essential for fault-tolerant quantum computing. High-rate quantum low-density parity-check (qLDPC) codes have emerged as a promising approach to reducing the substantial physical qubit overhead of quantum error correction, enabling encoding rates exceeding $1/2$ in blocks of approximately one thousand physical data qubits. However, their practical implementation requires decoders that combine high accuracy with low latency to keep pace with quantum hardware. Existing decoders of high-rate qLDPC codes can achieve high accuracy but often incur substantial computational latency, whereas neural network decoders have yet to be demonstrated for high-rate qLDPC codes. Here, we introduce \emph{neural triage decoding}, a framework that uses syndrome consistency to adaptively allocate neural computation. A local neural message-passing network rapidly predicts error corrections, additional iterations reduce residual syndrome violations, and a rarely invoked diversification stage searches for consistent corrections to the remaining shots. We demonstrate that the \texttt{Triage} decoder can decode ultra-high-rate qLDPC codes encoding a few hundred logical qubits, including lifted-product and Kasai codes, achieving state-of-the-art accuracy comparable to the near optimal algorithmic decoders while reducing decoding latency by approximately one to two orders of magnitude. Our framework is inspired by a rigorous analysis of the exact bitwise marginal decoder, revealing a distance-independent syndrome-inconsistent error floor. In addition, we show the robustness of neural network decoders to mismatched error priors. Our results establish adaptive neural decoding as a promising route towards fast, accurate, and scalable quantum error correction in the near term.

\end{abstract}

\maketitle
\raggedbottom

\let\realaddcontentsline\addcontentsline
\renewcommand{\addcontentsline}[3]{}

\section{Introduction}
\label{sec:introduction}

Quantum error correction (QEC), which protects quantum information through redundant encoding, is central to building a fault-tolerant quantum computer~\cite{Shor1995Scheme,Steane1996ErrorCorrection,gottesman2009introduction}. Recent experiments have demonstrated error suppression and key components of QEC on superconducting processors~\cite{Krinner2022Repeated,Google2023Suppressing,Google2025BelowThreshold}, neutral-atom platforms~\cite{Bluvstein2024Logical,Bluvstein2026Architecture}, and trapped-ion systems~\cite{Egan2021FaultTolerant,RyanAnderson2021RealTime,Postler2022FaultTolerant}, marking substantial progress towards fault-tolerant quantum computation. Scaling these advances to practical quantum algorithms~\cite{Shor1997Factoring,Lloyd1996Simulation,AspuruGuzik2005Chemistry,Grover1997Search,Harrow2009LinearSystems} requires sufficiently low logical error rates at a manageable physical-qubit cost~\cite{Gidney2021Factoring,Gidney2025Factoring,Webster2026Pinnacle,Beverland2022Assessing,Cain2026Shor}. Recent developments in high-rate quantum low-density parity-check (qLDPC) codes offer a route to reducing this cost by encoding many logical qubits within a single block while retaining sparse parity checks~\cite{Breuckmann2021QLDPC,Gottesman2014ConstantOverhead,Xu2024ConstantOverhead,Bhardwaj2026HighRateProcessors,Tillich2014PositiveRate,panteleev2022asymptotically,Leverrier2022Tanner,Bravyi2024HighThreshold,Kasai2026Orthogonality}. In particular, codes with encoding rates exceeding one half can accommodate hundreds of logical qubits in blocks of approximately one thousand physical data qubits, substantially reducing the encoding overhead~\cite{Zhao2026UltraHighRate,Kasai2026Orthogonality,Lee2026DesignPrinciples,Yang2026GALA,Lu2026UltraLowOverhead,Okada2026PairPartition}. Realizing this advantage also requires a fast and accurate classical decoder that infers corrections to logical observables from measured syndromes. Decoding accuracy determines how effectively the code suppresses logical errors, while sustained throughput must keep pace with syndrome generation and latency must meet the timing requirements of logical feedback~\cite{Fowler2012Processing,Terhal2015Memories,Google2025BelowThreshold}.

Exact decoding is computationally challenging, and finding an optimal correction is NP-hard even for general classical codes~\cite{Berlekamp1978Intractability}. Quantum degeneracy adds a further computational challenge, making optimal decoding of stabilizer codes $\#$P-complete~\cite{Iyer2015Hardness}. Practical decoders must therefore balance accuracy, speed, and versatility. Belief propagation (BP) is a natural choice for qLDPC codes because it exploits their sparse parity checks, supports parallel computation, and achieves near-optimal performance for classical LDPC codes~\cite{Gallager1962LDPC,Kschischang2001FactorGraphs,Chung2001Shannon}. In quantum codes, however, degeneracy and short cycles can prevent BP from converging, leaving corrections that do not reproduce the measured syndrome~\cite{MacKay2004SparseGraph,poulin2008iterative,Maan2026Correlated}. These difficulties have motivated BP variants such as relay-BP~\cite{Muller2025RelayBP}, as well as post-processing methods that supplement BP with ordered statistics decoding (OSD)~\cite{Fossorier1995OrderedStatistics,panteleev2021degenerate,Roffe2020Landscape}, localized statistics decoding (LSD)~\cite{Hillmann2025Localized}, or search-based correction~\cite{Zhao2026UltraHighRate,Beni2025Tesseract,Ye2026BeamSearch,Cain2024Correlated}. Such approaches can substantially improve accuracy and, in some settings, achieve near-optimal decoding performance~\cite{Zhao2026UltraHighRate}. Their computational cost, particularly for most-likely-error (MLE) search, nevertheless remains a barrier to the high throughput required for real-time decoding.

\begin{figure*}[htbp]
    \centering
    \includegraphics[width=0.99\linewidth]{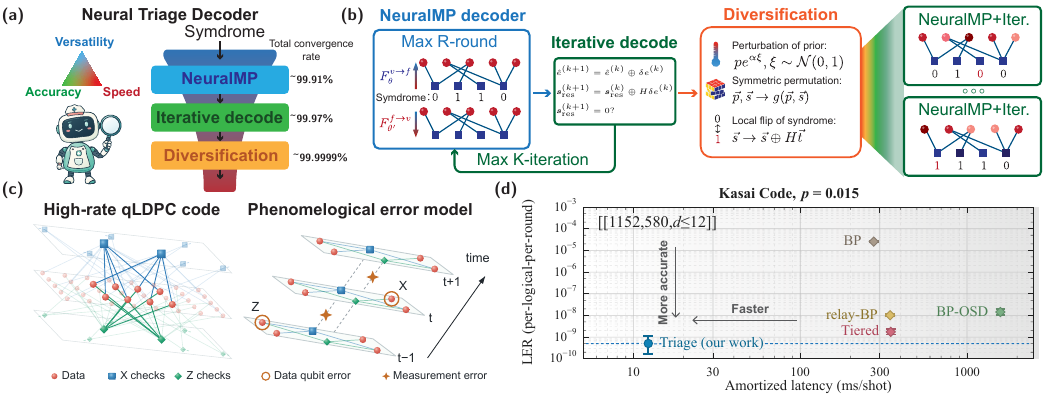}
    \caption{\textbf{Neural triage decoding for high-rate quantum LDPC codes.}
(\textbf{a}) The \texttt{Triage} decoder combines a fast NeuralMP local decoder with iterative decoding strategy. For rare cases when proposed errors cannot resolve all syndrome, it enters diversification stage. The \texttt{Triage} decoder allocate computation resources adaptively to achieve high accuracy, high speed and being versatile simultaneously.
(\textbf{b}) NeuralMP performs $R$ internal MP-rounds per evaluation. Then iterative decoding repeatedly applies the network to the residual syndrome for up to $K$ outer iterations. Remaining shots enter diversification stage, which uses prior perturbations, symmetric permutations and syndrome flips to explore alternative correction trajectories with the same trained network.
(\textbf{c}) High-rate qLDPC codes are represented by data qubits (red circles), $X$ checks (blue squares) and $Z$ checks (green diamonds). Under the phenomenological error model, data qubits experience Pauli $X$, $Y$ or $Z$ errors each per measurement round for multiple measurement rounds, while measurement outcomes are flipped with probability $p_{\mathrm m}$.
(\textbf{d}) Schematic comparison of decoding latency and logical error rate for neural and algorithmic decoders. Logical error rates are expressed on a per-qubit-per-round basis; the arrow indicates the desired direction of improvement.}
\label{fig:figure1}
\end{figure*}

Neural networks offer another route to the versatile, accurate and fast decoder envisioned in \Cref{fig:figure1}(a)~\cite{Torlai2017NeuralDecoder,maskara2019advantages,Senior2025AlphaQubit2}. They can learn directly from experimental data and adapt to device-specific noise~\cite{Baireuther2018Correlated,bausch2024learning,Lange2025GNN}. Autoregressive and transformer models can capture global correlations and approach optimal decoding, but often at substantial computational cost~\cite{cao2023qecgpt,bausch2024learning,BonillaAtaides2026Universal}. On the other hand, architectures such as Cascade neural decoders exploit code symmetries to achieve fast and accurate decoding for specific code families~\cite{gu2026scalable} and the performance for high-rate qLDPC codes is unclear. Graph neural networks (GNNs) offer broader applicability by using Tanner graphs as a common representation while retaining the sparse, parallel message passing of BP~\cite{liu2019neural,gong2023graph,ninkovic2024decoding,Maan2025Astra}. However, like BP, they can fail to produce syndrome-consistent corrections~\cite{gong2023graph}; whether they can match the strongest algorithmic decoders across different code families remains unclear. Although neural qLDPC decoders have been demonstrated on blocks encoding tens of logical qubits~\cite{ninkovic2024decoding,Maan2025Astra}, their performance on ultra-high-rate codes encoding hundreds of logical qubits per block remains unestablished. This regime demands accurate inference over a vast space of logical outcomes: a block encoding 500 logical qubits has $2^{500}\approx10^{150}$ logical classes in a single code block.

In this work, we introduce \emph{Neural Triage Decoding} (\texttt{Triage}), a versatile, fast, and accurate neural decoder that extends the idea of a local decoder through iterative decoding and adaptive search for syndrome-consistent corrections (see \Cref{fig:figure1} (a)). Its core graph-based message-passing decoder (\emph{NeuralMP}) predicts physical corrections through local updates on the sparse Tanner graph of a detector error model (DEM). Within each message-passing round (MP-round), updates to nodes of the same type can be evaluated in parallel, making neural inference well suited to GPU execution. Then it outputs a marginal likelihood of error on each error node. We also show the limitations of bitwise marginal decoding and propose \emph{iterative decoding} through virtual updates on the syndromes, shown as the second step in \Cref{fig:figure1} (b). If a syndrome compatible solution is still not reached, we invoke the last stage of \emph{diversification}, which combines syndrome flips~\cite{wang2025fullyparallelized}, perturbation of the prior distribution~\cite{poulin2008iterative}, and symmetric permutation~\cite{Koutsioumpas2025AutDEC}. The diversification reuses the same trained neural network as a general strategy to resolve syndrome-error compatibility without retraining.

We evaluate \texttt{Triage}'s performance on the high-rate lifted-product (LP) code $[[1122,148,d\leq 20]]$ (LP20)~\cite{Cain2026Shor} and ultra-high-rate Kasai code $[[1152,580,d\leq 12]]$~\cite{Kasai2026Orthogonality,Zhao2026UltraHighRate} with phenomenological error models across $r$ noisy measurement rounds, with a final perfect syndrome extraction. These two codes are representative of product-based and non-product-based constructions of high-rate qLDPC codes, respectively. While being versatile, \texttt{Triage} achieves state-of-the-art (SOTA) decoding accuracy that is comparable to or better than that of algorithmic decoders, such as BP-OSD and BP-relay-MLE (Tiered)~\cite{Zhao2026UltraHighRate}, achieving a per-qubit-per-round error rate of around $5\times10^{-10}$ at $p=1.5\%$ physical error rate on the Kasai code  (\Cref{fig:figure1}d). Alongside this accuracy, \texttt{Triage} reduces decoding latency and increases throughput by factors of $10$-$100$ relative to algorithmic decoders (\Cref{fig:figure1}d). \texttt{Triage} thus brings versatility, accuracy, and speed together in a common framework for decoding ultra-high-rate qLDPC codes. With further hardware optimization, this framework offers a path towards real-time error correction of hundreds of logical qubits within a single code block.

\section{Results}
\label{sec:results}

\subsection*{Local marginal decoder and iterative decoding}
\label{subsec:local_decoder}

\begin{figure*}[htbp]
    \centering
    \includegraphics[width=0.99\linewidth]{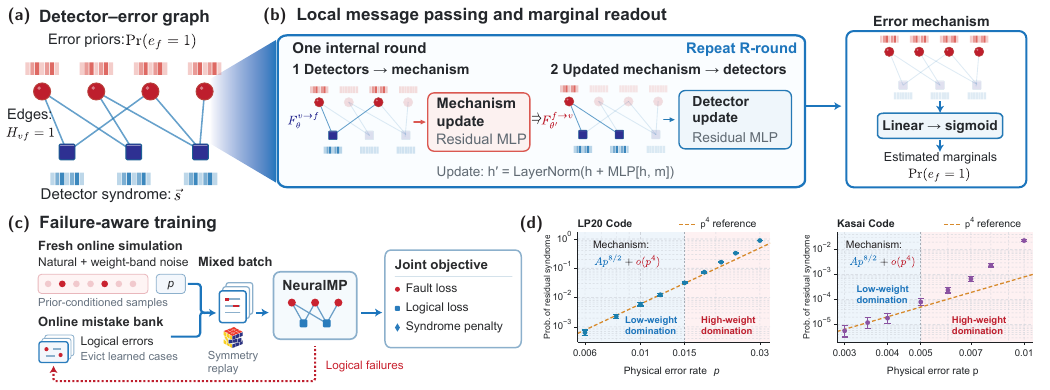}
    \caption{\textbf{NeuralMP architecture, failure-aware training and residual-syndrome scaling.}
(\textbf{a}) The Tanner graph of the DEM connects fault mechanisms (red circles) to detectors (blue squares) through $H_{vf}=1$. Red colored bars represent node states encoding mechanism priors together with positions embedding and blue colored bars represent detector outcomes embedding.
(\textbf{b}) Each MP-round first updates fault states using neighbouring detector states, then updates detector states using the updated fault states. Aggregated messages enter residual multilayer perceptrons (MLPs) with layer normalization. After $R$ MP-rounds, a linear readout followed by a sigmoid produces one marginal estimate per mechanism.
(\textbf{c}) Training combines fresh online simulations, including error-weight-conditioned samples, with symmetry-augmented replay from an evolving mistake bank. Logical failures enter the bank, while repeatedly corrected examples are removed. The joint objective supervises fault predictions, logical observables and syndrome consistency.
(\textbf{d}) Probability of a nonzero residual syndrome after a single pass of NeuralMP versus physical error rate $p$ for the LP20 and Kasai codes. Orange dashed lines indicate $p^4=p^{w_{\min}/2}$ reference scaling with $w_{\min}$ as minimal weight of space-time degeneracy. Pale blue and pale red regions highlight regimes dominated by low-weight and high-weight errors, respectively. Error bars show 95\% confidence intervals.}
\label{fig:figure2}
\end{figure*}
We formulate decoding using a detector error model with $F$ fault mechanisms and $m$ detectors. A fault configuration $\vec{e}$ produces a syndrome $\vec{s}=H\vec{e}$ and logical-observable flips $\boldsymbol{\ell}=L\vec{e}$, with arithmetic over $\mathrm{GF}(2)$. NeuralMP predicts one score per fault mechanism by passing messages on the sparse bipartite graph defined by $H$~\cite{ninkovic2024decoding}, which is also called the Tanner graph. Syndrome values $s_v$ and mechanism priors $p_f$ are encoded as $\chi$-dimensional detector and fault states, shown as colored bars in \Cref{fig:figure2}(a), with learned coordinate embeddings supplying positional information. Each of $R$ MP-rounds first updates fault states using detector states, then updates detector states using the updated fault states (\Cref{fig:figure2}(b)). Trainable functions $F_\theta^{v\to f}$ and $F_{\theta'}^{f\to v}$ combine each node's current state with the sum of its incoming messages. Parameters are shared across nodes of each type within an MP-round, enabling parallel updates, but differ between MP-rounds. A linear readout of the final fault states gives logits $z_f$ and learned fault scores $q_f=\sigma(z_f)$, where $\sigma(x)=(1+e^{-x})^{-1}$ is the sigmoid function. Crucially, for $N=m+F$ nodes and $E$ edges, one evaluation costs $\mathcal{O}(R[E\chi+N\chi^2])$, which is linear in $N$ when $E=\mathcal{O}(N)$ at fixed $R,\chi$. In contrast, the global decoders that use attention mechanism will have costs that scale as $\mathcal{O}(N^2)$~\cite{vaswani2017attention}.

During training, the \texttt{Triage} decoder readily learns to correct common error patterns, but rare and difficult decoding cases are the important ones determining its performance in the low error rate regime. To enhance the sparse training reward, we couple online error-weight-conditioned simulation with a dynamically updated mistake bank (\Cref{fig:figure2}(c)). We store hard decoding examples in the online mistake bank and remove them from the bank after repeated successes. We also augment replayed examples using code symmetries, jointly permuting fault patterns and syndromes. This failure-aware loop combines targeted data generation with the supervision of fault predictions, logical observables, and syndrome consistency. Implementation details are provided in the Methods.

Independent decisions from marginal probabilities can violate syndrome constraints even when the marginals are \emph{exact}. To illustrate the ambiguity, consider independent $Z$ noise on $n$ qubits of a Calderbank--Shor--Steane (CSS) code and an even-weight $Z$ stabilizer $g$ of weight $w$. Errors $e_A$ and $e_B=e_A\oplus g$ occupying complementary halves of its support both have weight $w/2$. Since $Hg=0$ and $Lg=0$, they share the same syndrome and logical labels. Either specified error, with no other errors, occurs with probability $p^{w/2}(1-p)^{n-w/2}=\Theta(p^{w/2})$. If this pair gives the only minimum-weight explanations of its syndrome, each supported qubit has an exact posterior marginal approaching $1/2$ as $p\to0$. Such local ambiguities can arise from stabilizers whose weight does not grow with code distance. In phenomenological noise models, space--time degeneracies involving both data and measurement errors can have lower weight than purely spatial stabilizers.

\begin{theorem}[Syndrome inconsistency under phenomenological noise]
\label{thm:marginal-residual}
Consider a CSS code with Z-check matrix $H_Z$ of minimum qubit degree $\nu\geq2$, with no 4-cycles in its Tanner graph and no nonzero kernel vector of weight below $2\nu+2$. In the independent $X/Z$ phenomenological model defined in Methods, with $r\geq5$ noisy measurement rounds, data error rate $p$, and measurement error rate $q=\Theta(p)$, the exact bitwise marginal decoder, using maximum-a-posteriori (MAP) decisions, satisfies
\begin{equation}
\begin{aligned}
P_{\mathrm{inc}}(p):=\Pr(H\hat e\ne He)&=\Omega(p^{\kappa(\nu)})\qquad(p\to0),\\
\kappa(\nu)&=\begin{cases}\nu/2+1,&\nu\text{ even},\\\nu+1,&\nu\text{ odd}.\end{cases}
\end{aligned}
\label{eq:marginal_residual_scaling}
\end{equation}
\end{theorem}
The LP and Kasai code shots used here have $\nu=3$, so Theorem~\ref{thm:marginal-residual} yields $P_{\mathrm{inc}}(p)=\Omega(p^4)$ for their independent $X/Z$ phenomenological models. Single-pass NeuralMP results are consistent with this mechanism (\Cref{fig:figure2}(d)). The proof and finite-instance certificates are given in the Methods and the Supplementary Information.

Motivated by this understanding, we propose the idea of \emph{iterative decoding}, which addresses residual violations by repeatedly applying the same trained network to the current residual syndrome (\Cref{fig:figure1}(b)). Starting from $\hat e^{(0)}=0$ and $\boldsymbol{s}_{\mathrm{res}}^{(0)}=s$, each outer iteration $k$ samples a proposal $\delta e^{(k)}$ and updates
\begin{align}
\hat e^{(k+1)}&=\hat e^{(k)}\oplus\delta e^{(k)},\nonumber\\
\boldsymbol{s}_{\mathrm{res}}^{(k+1)}&=\boldsymbol{s}_{\mathrm{res}}^{(k)}\oplus H\delta e^{(k)}.
\label{eq:ntd_iterative_updates}
\end{align}
Decoding stops when the residual vanishes or reaches the maximum $K_{\text{max}}$ iterations. Practically, we set $K_{\text{max}}=4$. Shots unresolved after $K_{\text{max}}$ outer iterations proceed to the final stage of diversification.

\subsection*{Improved convergence with diversification}
\label{sec:results_2}

Residual-guided iteration can be complemented by diversification, which explores alternative correction trajectories under different inference conditions. The relay-BP decoder follows the same principle by varying local memory strengths to diversify BP trajectories~\cite{Muller2025RelayBP}. For a fixed neural decoder, we generate alternative correction candidates through prior perturbations, symmetric permutations, and syndrome flips, using information from the decoder to guide the search (see \Cref{fig:figure1} (b)). In particular, near-zero logit outputs identify uncertain error mechanisms, while repeated corrections at the same locations indicate instability along an iterative trajectory.

\begin{figure}[htbp]
    \centering
    \includegraphics[width=0.99\linewidth]{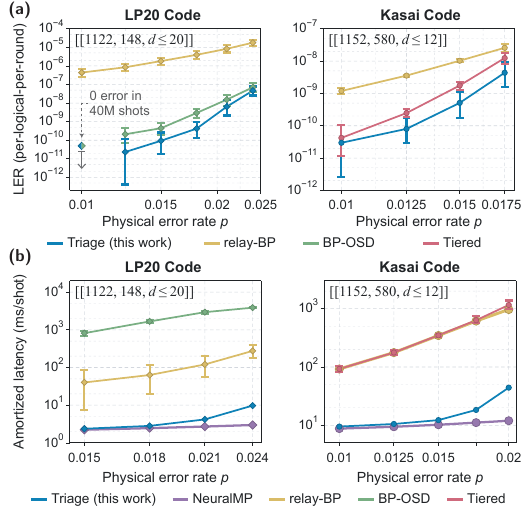}
    \caption{\textbf{Decoding accuracy and amortized latency on LP20 and Kasai codes.}
(\textbf{a}) Logical error rate (LER) per-qubit-per-round versus physical error rate $p$, using approximately $0.04$--$40.96$ million shots per point for LP20 code and $0.24$--$20.48$ million for Kasai code. Error bars quantify finite-sampling uncertainty using 95\% Clopper--Pearson intervals.
(\textbf{b}) Amortized decoding latency. We benchmark the latency using NVIDIA H200 GPU. NeuralMP (\texttt{Triage}'s frontend) reports the asymptotic frontend cost, whereas \texttt{Triage} shows the composed full-decoder cost at batch size 64, including handoff and diversification. CPU baseline timing uses 1,000--5,000 shots per point. Timing error bars, where shown, are 95\% bootstrap confidence intervals from 2,000 resamples and indicate uncertainty in the estimated mean latency. CPU baselines use one thread, except the eight-thread COPT stage~\cite{Ge2023COPT}.
}
    \label{fig:figure3}
\end{figure}

Each diversification trajectory starts from the original syndrome and reuses the iterative decoder under modified inference conditions. Prior perturbations modify the mechanism priors supplied to the network, shifting its preferences among competing correction proposals. Symmetric permutations preserve the Tanner graph of the DEM and mechanism priors but change the alignment of the syndrome with the network's coordinate encodings, allowing different predictions that are mapped back to the original coordinates. For syndrome flips, we select a small set of mechanisms, represented by a binary vector $\vec{t}$, using marginal uncertainty and repeated corrections as guides. We apply the iterative decoder to the modified syndrome $\vec{s}\oplus H\vec{t}$ with the original priors, then combine its prediction $\vec{e'}$ with $\vec{t}$ to obtain the candidate $\hat e=\vec{t}\oplus \vec{e'}$. The resulting candidates are checked against the original syndrome, and syndrome-consistent candidates are ranked by their negative log-likelihood under the original mechanism priors.

The effect of a syndrome flip can be seen in the half-support issue described above, where two complementary errors $\vec{a}$ and $\vec{b}$ each have weight $w/2$ and produce the same syndrome. Choosing $\vec{t}$ to contain a single qubit in the support of $\vec{a}$ maps these errors to $\vec{a}\oplus \vec{t}$ and $\vec{b}\oplus \vec{t}$, both consistent with the modified syndrome but with weights $w/2-1$ and $w/2+1$, respectively. Under the unchanged independent-error prior with $p<1/2$, the former is more likely by a factor $[(1-p)/p]^2$. In the two-candidate example, this imbalance moves the marginals away from $1/2$ and favors a consistent choice, which is then mapped back to a correction for the original syndrome. The flip thus changes the preference between candidates during search without changing their relative likelihood in the original decoding problem.

We organize these diversification strategies into a conditional decoding procedure that allocates additional computation to shots with unresolved syndromes. Shots whose residual syndrome vanishes during the initial iterative stage return immediately, and the remainder enter a syndrome-flip stage guided by marginal uncertainty and repeated corrections. If it produces syndrome-consistent candidates, we return the best candidate according to the selection rule above. Shots with no such candidate can proceed to a supplementary search stage which has further sign flip strategy and symmetric permutations combined with prior perturbations. Predictions from different symmetry-related views are mapped back to the original coordinates, and their disagreement identifies locations for additional syndrome flips. Candidates from this stage are evaluated by the same rule. If the search budget is exhausted without a syndrome-consistent candidate, the decoder retains the correction returned by the initial iterative stage, and the shot remains unresolved.

\subsection*{Decoding high-rate qLDPC codes}
\label{sec:results_3}
We evaluate \texttt{Triage}'s performance on the lifted-product code $[[1122,148,\leq20]]$ and the ultra-high-rate Kasai code $[[1152,580,d\leq12]]$, spanning encoding rates of $13.2\%$ and $50.3\%$, respectively. Both benchmarks use phenomenological noise, with data errors and multiple noisy syndrome measurements followed by a final perfect syndrome extraction. For the LP20 code, we simulate five noisy measurement rounds, with independent $X$ and $Z$ errors, each occurring with probability $p$ on every data qubit, and measurement outcomes flipped with probability $p$. For the Kasai code, we simulate eight noisy measurement rounds under depolarizing noise, with $X$, $Y$, and $Z$ errors, each occurring with a probability of $p/3$, and measurement outcomes flipped with a probability of $4p/3$. For both codes, diversification begins with syndrome flips. For LP20, unresolved shots proceed to an additional stage combining prior perturbations, symmetric permutations, and further syndrome flips guided by disagreement between symmetry-related predictions; Kasai uses only the initial syndrome-flip stage. Crucially, for both LP20 and Kasai codes, we train \texttt{Triage} at a single physical error rate, $p=2\%$, and evaluate its performance across a range of previously unseen error rates.

We observe that \texttt{Triage} decoder achieves low logical error rates while substantially reducing amortized decoding costs. For example, at $p=0.015$, its estimated per-qubit-per-round LER is $9.33\times10^{-11}$ on LP20 and $5.08\times10^{-10}$ on Kasai, compared with $4.43\times10^{-10}$ for BP-OSD and $1.76\times10^{-9}$ for the BP-relay-MLE Tiered decoder, respectively (see \Cref{fig:figure3} (a)). These accuracies are obtained through iterative NeuralMP decoding with diversification. At the same error rate, the estimated full-decoder costs at batch size 64 are $2.37$~ms and $12.3$~ms per shot on a single NVIDIA H-series GPU, compared with CPU baseline means of $815$~ms and $349$~ms for BP-OSD and the Tiered decoder, respectively (see \Cref{fig:figure3} (b)). These benchmark settings yield approximately $343\times$ and $28\times$ speedups.

We further examine its robustness to inaccurate noise calibration by perturbing the priors supplied to the decoder while keeping the physical noise fixed. This is important as the noise distribution can shift from the initial noise learning, and control drifts can also happen during the computation \cite{Sivak2026Reinforcement,Gong2026SelfCalibrating}. We perturb the prior error uniform distribution as $p_f\rightarrow p_f(1+\sigma_r\mathcal{N}(0,1))$, where $\mathcal{N}(0,1)$ is the Gaussian distribution with zero mean and unit variance. We investigate the ratio between the whole block logical error rate with the perturbed prior distribution over that of the original prior distribution. In \Cref{fig:figure4}, we show that the four-iteration \texttt{Triage} frontend (NeuralMP) maintains a block-LER ratio close to unity across the tested perturbations. At the largest nominal mismatch strength, $\sigma_r=0.5$, its ratios are approximately $1.01$ for the LP20 code and $0.98$ for the Kasai code, whereas the corresponding algorithmic baselines show increases of approximately $1.6$--$2.1$ times relative to matched-prior decoding. These results suggest that neural-network-based decoders can offer passive resilience to uncertainty and fluctuations in noise parameters, analogous to robust quantum controls \cite{doi:10.1126/sciadv.adr0875,2025arXiv250819075H,PhysRevLett.132.193801}.

\begin{figure}
    \centering
    \includegraphics[width=0.99\linewidth]{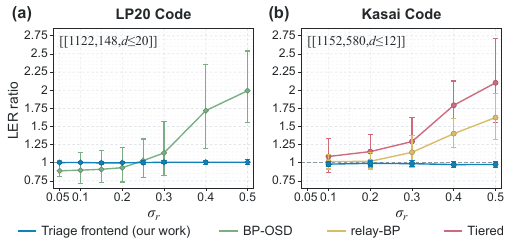}
    \caption{\textbf{Robustness to mismatched error priors.}
Block-LER ratio relative to correctly specified priors versus nominal mismatch strength $\sigma_r$ for (\textbf{a}) LP20 code at $p=0.024$, and (\textbf{b}) Kasai code at $p=0.0175$. NeuralMP uses four decoding iterations. Each prior realization is evaluated on $9\times10^5$ shots for LP20 code and $6\times10^5$ shots for Kasai code, sharing the same noise samples with the matched-prior baseline. Two prior realizations are used per $\sigma_r$, except LP20 at $\sigma_r=0.1$ and $0.3$, which use four. Markers show geometric means of the per-realization LER ratios. Error bars span the per-realization ratio intervals obtained by transforming conditional 95\% Clopper--Pearson intervals.
}
    \label{fig:figure4}
\end{figure}

\section{Discussion}
\label{sec:discussion}
\texttt{Triage} decoder combines a local neural message-passing network with iterative decoding and a rarely invoked diversification to allocate computation according to decoding difficulty. This architecture is inspired by the theoretical understanding of the limitations of the exact bitwise marginal decoder, where there exists a lower bound on the residual syndrome density even when the marginal distribution is exact. The results on lifted-product and Kasai codes show that \texttt{Triage} can achieve SOTA decoding accuracy comparable to and better than that of the algorithmic decoders, with substantially lower amortized decoding costs than the tested algorithmic baselines. In addition, we show the robustness of the \texttt{Triage} decoder against the error prior distribution, which is important and relevant for real-time decoding with hardware fluctuations. These results also suggest several directions for improving the decoder and connecting it to experimental operation.

Symmetric permutations currently enter \texttt{Triage} through diversification. Incorporating these symmetries directly into NeuralMP, for example, through symmetry-respecting parameter sharing as suggested in Cascade~\cite{gu2026scalable}, could reduce network size and inference cost. The different effectiveness of diversification proposals across the two codes also motivates a more systematic analysis of decoding failures. Relating the response to syndrome flips, prior perturbations, and symmetric permutations to graph structure and degeneracy could reveal which ambiguities each proposal resolves and guide the choice of proposal and computational budget for each shot.

These diagnostics could support recursive self-improvement (RSI) of \texttt{Triage}. Automated discovery of programs and algorithms with language models~\cite{RomeraParedes2024FunSearch,Novikov2025AlphaEvolve}, together with the discovery of learning rules through meta-learning~\cite{Oh2025DiscoRL}, provides precedents for using machine-generated proposals and quantitative feedback to improve computational methods. In this setting, AI agents could analyze decoding failures, propose changes to the architecture, training procedure, or diversification strategy, and evaluate successive revisions against logical error rates and total decoding costs. This feedback loop could reduce manual tuning while retaining explicit accuracy and runtime criteria \cite{meta_physical_design}.

Finally, adaptive decoding and hardware self-calibration could form a coupled feedback loop. Neural decoders can be adapted to experimental data~\cite{bausch2024learning}, while recent experimental and theoretical work shows how syndrome-derived signals can guide control optimization during error correction~\cite{Sivak2026Reinforcement,Gong2026SelfCalibrating}. \texttt{Triage} could update mechanism priors or selectively retrain NeuralMP as the noise changes, while detection-event statistics guide the recalibration of the underlying controls. The observed resilience of the four-iteration frontend to the tested prior perturbations provides motivation for studying this integration. Establishing its ability to track physical drift would require evaluating the full decoder and calibration loop together, including the cost and stability of online updates.

Extending \texttt{Triage} from phenomenological quantum memory to circuit-level logical gadgets and algorithms~\cite{Cohen2022CodeSurgery,Zhou2026ParallelAlgebraicSurgery,Yoder2025TourDeGross,He2025Extractors,Bluvstein2024Logical,Cain2024Correlated,Zhou2025TransversalFaultTolerance}, or applying it to experimental data, are central next steps. Real-time operation also requires measuring single-shot latency, its distribution, and sustained throughput against the timing requirements of the hardware~\cite{Google2025BelowThreshold,Senior2025AlphaQubit2}. Further work is needed to optimize GPU inference and investigate field-programmable gate array (FPGA) deployment of the trained model~\cite{NVIDIACUDAQNeuralDecoders,Cicero2026GNNAccelerator}. The present GPU and CPU comparisons quantify amortized costs under their stated batching and hardware conditions. These further tests would determine whether adaptive neural decoding can sustain both accurate logical protection and efficient feedback during prolonged operation.

\section{Methods}
\label{sec:methods}

\subsection{NeuralMP architecture}
\label{sec:methods_1}
\label{methods:neuralmp}

The detector error model specifies $F$ mechanisms, $m$ detectors, and $k$ logical observables through $H\in\{0,1\}^{m\times F}$, $L\in\{0,1\}^{k\times F}$, and priors $p_f$, with $N=m+F$ nodes in a bipartite Tanner graph. Error configurations differing by a vector in $\ker H\cap\ker L$ share both syndrome and logical labels. Initial states combine learned node-type and coordinate features; fault states additionally receive a linear projection of $\log p_f$, and detector states receive the encoded binary outcomes. Coordinate features concatenate learned embeddings for individual nodes, followed by a linear projection and a rectified linear unit (ReLU). Messages include a learned source-type bias before summation, and detector aggregation uses the newly updated fault states. Each directional update takes the current state and aggregate as the input to a two-layer MLP with a ReLU hidden activation. Its output is added to the current state and layer-normalized. Distinct update modules are used for the two directions and for successive MP-rounds, with each module applied in parallel across its node family. The deployed Kasai network uses hidden dimension $\chi=64$ and $R=80$ MP-rounds with distinct MP-round-specific parameters.

Fresh training examples mix natural noise with noise conditioned on selected fault weights. Thresholded predictions with incorrect logical observables enter the replay bank. Verified symmetries transform replayed fault patterns and syndromes together, with logical labels recomputed. The three loss terms are defined in the Training objective subsection. Class weighting and modified sampling distributions mean that the scores are not assumed to be calibrated deployment posteriors.

Each outer iteration is an independent forward pass: hidden states are re-initialized from the current residual syndrome and the original priors, and no hidden state or logit is carried over between iterations. Information propagates between iterations only through the residual syndrome. Proposal bits are sampled independently as $\delta e_f^{(k)}\sim\operatorname{Bernoulli}[\sigma(z_f^{(k)}/T)]$ and accumulated according to \Cref{eq:ntd_iterative_updates}. Here $T>0$ is the sampling temperature. The deployed configuration uses $T=1$, so $q_f^{(k)}=\sigma(z_f^{(k)})$, and a budget of at most $K=4$ network evaluations per shot before escalation to diversification. Parameters remain fixed throughout inference.

\subsection{Training objective}

We train NeuralMP using physical-error-level, logical, and syndrome supervision. Let $q_f=\sigma(z_f)$ denote the network output for fault mechanism $f$, and let $e$, $s=He$ and $\boldsymbol{\ell}=Le$ denote the simulated fault pattern, syndrome and logical-observable flips, respectively. All parity operations are over $\mathrm{GF}(2)$. The total loss is
\eqs{
\mathcal{L}
=\mathcal{L}_{\mathrm{fault}}
+\mathcal{L}_{\mathrm{logical}}
+\mathcal{L}_{\mathrm{syndrome}},
}
with unit coefficients for all three terms.

The fault loss is binary cross-entropy (BCE), with a positive-class weight of eight to compensate for sparse faults:
\eqs{
\mathcal{L}_{\mathrm{fault}}
=-\frac{1}{F}\sum_f
\left[8e_f\log q_f+(1-e_f)\log(1-q_f)\right].
}
For logical supervision, we construct a differentiable parity estimate under a factorized Bernoulli approximation,
\eqs{
\widetilde{\ell}_j
=\frac{1-\prod_{f:L_{jf}=1}(1-2q_f)}{2},\\
\mathcal{L}_{\mathrm{logical}}
=\frac{1}{k}\sum_j
\mathrm{BCE}(\widetilde{\ell}_j,\ell_j)
}

The syndrome-consistency penalty similarly compares predicted fault parities with the measured detectors:
\eqs{
\mathcal{L}_{\mathrm{syndrome}}
=\frac{1}{m}\sum_c
\left[
1-(-1)^{s_c}
\prod_{f:H_{cf}=1}(1-2q_f)
\right].
}
Here, $F$, $k$, and $m$ are the numbers of fault mechanisms, logical observables, and detectors. These terms encourage accurate fault predictions, correct logical outcomes and syndrome consistency, respectively; exact consistency is checked separately during decoding.

\subsection{Marginal decoding and syndrome inconsistency}
\label{sec:methods_2}
We first establish a criterion for marginal decisions to violate a syndrome constraint. We then derive this criterion from the spatial graph and phenomenological noise assumptions of Theorem~\ref{thm:marginal-residual}, and verify those assumptions for the simulated shots. Throughout this subsection, $H\in\mathbb F_2^{m\times F}$ is fixed and finite, and the error mechanisms are independent with probabilities $p_f(p)=\Theta(p)$ as $p\to0$. The exact posterior marginals are $\pi_f(s)=\Pr(e_f=1\mid He=s)$. Bitwise MAP sets $\hat e_f(s)=\boldsymbol{1}\{\pi_f(s)>1/2\}$, assigning ties to zero. All matrix operations are over $\mathbb F_2$. We write $P_{\mathrm{inc}}(p)=\Pr(H\hat e\ne He)$.

\begin{lemma}[Concentration on minimum-weight explanations]
\label{lem:methods-minimum-concentration}
For a nonzero attainable syndrome $s_*$, let $\kappa$ be its minimum compatible error weight, $\mathcal M_*$ the set of all compatible weight-$\kappa$ errors, and $U_*$ the union of their supports. Then
\begin{equation}
\begin{aligned}
Z_*:=\Pr(He=s_*)&=\Theta(p^\kappa),\\
\Pr(e\notin\mathcal M_*\mid He=s_*)&=O(p).
\end{aligned}
\label{eq:methods-minimum-concentration}
\end{equation}
In particular, $\pi_f(s_*)=O(p)$ for every $f\notin U_*$.
\end{lemma}
\begin{proof}
Independence and $p_f(p)=\Theta(p)$ imply that each fixed weight-$t$ error has probability $\Theta(p^t)$. Since $H$ is finite, the total probability of the minimum-weight explanations is $\Theta(p^\kappa)$, while that of all higher-weight compatible errors is $O(p^{\kappa+1})$. Dividing by $Z_*$ proves posterior concentration. An error on a coordinate outside $U_*$ can occur only in a higher-weight explanation, giving the marginal bound.
\end{proof}
For an equally likely disjoint pair $\mathcal M_*=\{a,b\}$, this lemma gives $\pi_f(s_*)=1/2+O(p)$ on $U_*$. The higher-weight explanations determine the signs of these small shifts, which control the hard decisions.

\begin{lemma}[Two-bit marginal shifts]
\label{lem:methods-endpoint-shifts}
Fix a syndrome $s_*$ with $Z_*>0$ and distinct bits $i,j$. Define $W_{uv}(p)=\Pr(He=s_*,e_i=u,e_j=v)$ for $u,v\in\{0,1\}$, $\Delta=W_{00}-W_{11}$, and $\delta=W_{10}-W_{01}$. Then
\begin{equation}
\begin{aligned}
\pi_i(s_*)-\tfrac12&=\frac{-\Delta+\delta}{2Z_*},\\
\pi_j(s_*)-\tfrac12&=\frac{-\Delta-\delta}{2Z_*}.
\end{aligned}
\label{eq:methods-marginal-shifts}
\end{equation}
The two marginals lie strictly on the same side of $1/2$ if and only if $|\Delta|>|\delta|$. In particular, $\Delta>|\delta|$ places both below $1/2$.
\end{lemma}
\begin{proof}
Use $\pi_i=(W_{10}+W_{11})/Z_*$, $\pi_j=(W_{01}+W_{11})/Z_*$, and $Z_*=\sum_{u,v}W_{uv}$. Subtracting $1/2$ gives the identities. Their product is $(\Delta^2-\delta^2)/(4Z_*^2)$, which is positive exactly when the shifts have the same nonzero sign.
\end{proof}
\begin{proposition}[A two-bit inconsistency criterion]
\label{prop:methods-marginal-criterion}
Under the noise assumptions of this subsection, suppose a nonzero syndrome $s_*$ has exactly two minimum-weight explanations $a,b$ of weight $\kappa$ with disjoint supports, and set $U=\operatorname{supp}(a\oplus b)$. Suppose every compatible error supported on $U$ obeys $e_i\oplus e_j=1$ for fixed distinct bits $i,j\in U$. If $|\Delta|>|\delta|$ for these bits at all sufficiently small $p>0$, then exact bitwise MAP has $P_{\mathrm{inc}}(p)=\Omega(p^\kappa)$.
\end{proposition}
\begin{proof}
By Lemma~\ref{lem:methods-minimum-concentration}, $Z_*=\Theta(p^\kappa)$ and all marginals outside $U$ are $O(p)$. Hence the MAP output is supported on $U$ for sufficiently small $p$. Lemma~\ref{lem:methods-endpoint-shifts} gives equal MAP decisions on $i,j$, violating the required odd parity. Thus MAP is inconsistent whenever $s_*$ is observed, and $P_{\mathrm{inc}}(p)\geq Z_*=\Theta(p^\kappa)$, independently of ties at other coordinates.
\end{proof}
The quantities $\Delta$ and $\delta$ separate a common shift of the two endpoint marginals from a preference between them: $\pi_i+\pi_j=1-\Delta/Z_*$ and $\pi_i-\pi_j=\delta/Z_*$. The phenomenological construction below gives these two effects a concrete space--time interpretation.

Theorem~\ref{thm:marginal-residual} assumes mutually independent $X/Z$ data faults (probability $p$ each) and measurement flips ($q=\Theta(p)$), with known initial syndromes and a perfect final readout after $r$ noisy measurement rounds. The code and $r$ are fixed as $p\to0$.

Choose a degree-$\nu$ qubit $v_*$. A data fault on this qubit can be moved by $\tau$ measurement rounds without changing the syndrome if a string of measurement faults joins the two times on each neighboring check. Writing $\vec d_t$ for its $X$ fault in measurement round $t$ and $\vec m_{h,t}$ for a measurement fault connecting detectors in measurement rounds $t,t+1$ on check $h$, this relation is
\[
\vec g=\vec d_0+\vec d_\tau+
\sum_{h\in N(v_*)}\sum_{t=0}^{\tau-1}\vec m_{h,t},
\qquad |\vec g|=2+\tau \nu.
\]
Here $H\vec g=L\vec g=0$. We split $U=\operatorname{supp}(\vec g)$ into two complementary explanations $a,b$, each containing one data endpoint and half of the measurement faults. Taking $\tau=1$ for even $\nu$ and $\tau=2$ for odd $\nu$ gives weight $\kappa(\nu)=1+\tau\nu/2$. The endpoint bits are $i=d_0$ and $j=d_\tau$: they refer to the same qubit at different times. The two explanations have equal prior probability even when $p\ne q$, and contribute equally to $W_{10}$ and $W_{01}$.

Higher-weight explanations modify this balance in two ways. Redistributing data faults on $v_*$, with compensating measurement strings, can produce endpoints $00$ or $11$ and hence change $\Delta$. For $\nu=3$, the first such pair uses either one data fault at the intermediate time or data faults at all three times, together with four or two measurement faults, respectively. A data fault beyond the endpoint pair gives an additional $00$ explanation. Moving the construction from the initial boundary into the interior allows this extension on both sides.

The leading $10,01$ extensions involve a pair of data faults at adjacent times on a neighboring degree-$\nu$ qubit. Together with compensating edges on that qubit's other checks, this pair replaces an occupied measurement edge. These extensions preserve the original endpoint values $10$ or $01$; unequal numbers of available extensions produce the leading contribution to $\delta$.

For $\nu=3$, label the neighboring checks $1,2,3$ and choose
\begin{equation}
\begin{aligned}
\vec a&=\vec d_0+\vec m_{1,0}+\vec m_{1,1}+\vec m_{2,1},\\
\vec b&=\vec d_2+\vec m_{2,0}+\vec m_{3,0}+\vec m_{3,1}.
\end{aligned}
\label{eq:methods-spacetime-halves}
\end{equation}
Their common syndrome activates checks $\{2,3\}$, $\{2\}$, and $\{1,2\}$ in measurement rounds $0,1,2$. More generally, for even $\nu=2k$, partition the neighboring checks into sets $\mathcal A,\mathcal B$ of size $k$, and take
\[
\vec a=\vec d_0+\sum_{h\in \mathcal A}\vec m_{h,0},\qquad
\vec b=\vec d_1+\sum_{h\in \mathcal B}\vec m_{h,0}.
\]
For odd $\nu=2k+1$, partition them into $\mathcal A,\{c\},\mathcal B$ with $|\mathcal A|=|\mathcal B|=k$, and use
\[
\begin{aligned}
\vec a&=\vec d_0+\sum_{h\in \mathcal A}(\vec m_{h,0}+\vec m_{h,1})+\vec m_{c,1},\\
\vec b&=\vec d_2+\sum_{h\in \mathcal B}(\vec m_{h,0}+\vec m_{h,1})+\vec m_{c,0}.
\end{aligned}
\]
For either degree parity, write $\vec s_*=H\vec a=H\vec b$. Summing the detector equations over time requires $e_i\oplus e_j=1$ for every compatible correction supported on $U$. Once the data endpoint is chosen, the measurement equations fix all remaining bits, so $a,b$ are the only compatible corrections on $U$.

The spatial assumptions in Theorem~\ref{thm:marginal-residual} control explanations extending outside $U$. The absence of Tanner 4-cycles limits distinct qubits to at most one shared check, while the bound on $\ker H_Z$ excludes spatial zero-syndrome patterns of weight below $2\nu+2$, including stabilizers. Together, these conditions establish that $a,b$ are the only minimum-weight explanations and allow the higher-weight contributions to $\Delta,\delta$ to be classified, as stated next and proved in Supplementary Information.

\begin{proposition}[Two-rate structure of the phenomenological witnesses]
\label{prop:methods-phenom-structure}
Under the spatial and measurement-round assumptions of Theorem~\ref{thm:marginal-residual}, the constructed $\vec a,\vec b$ are the only minimum-weight explanations of $\vec s_*$, with weight $\kappa=\kappa(\nu)$. Each contains one data fault and $\kappa-1$ measurement faults, so they remain equally likely for arbitrary $p,q$. Put $x=p/(1-p)$, $y=q/(1-q)$, and let $F_0=1+O(p)$ be the common no-error factor of the selected sector times the probability of zero syndrome in the other sector. Let $n_h$ count the other degree-$\nu$ qubits at check $h$, and put $N_A=\sum_{h\in\mathcal A}n_h$, $N_B=\sum_{h\in\mathcal B}n_h$.
For even $\nu$, the endpoint imbalances obey
\begin{equation}
\begin{aligned}
\Delta&=F_0xy^{3\nu/2}+O(p^{\kappa+\nu+1}),\\
\delta&=F_0(N_A-N_B)x^3y^{3\nu/2-2}\\
&\hspace{1em}+O(p^{\kappa+\nu+1}).
\end{aligned}
\label{eq:methods-ind-imbalance}
\end{equation}
For odd $\nu$, use two placements of the construction: $j=0$ starts in measurement round zero, while $j=1$ shifts every fault and detector by one measurement round. Both have the same minimum weight and prior composition. Their endpoint imbalances satisfy
\begin{equation}
\begin{aligned}
\Delta_j&=F_0\bigl[xy^{\nu-1}(y^2-x^2)\\
&\hspace{3em}+(1+j)xy^{2\nu}\bigr]+O(p^{2\nu+2}),\\
\delta_j&=2F_0(N_A-N_B)x^3y^{2\nu-2}\\
&\hspace{1em}+O(p^{2\nu+2}),\qquad j=0,1.
\end{aligned}
\label{eq:methods-two-rate-imbalance}
\end{equation}
\end{proposition}
\begin{proof}
In Supplementary Information, it classifies all compatible configurations through weight $\kappa+\nu$, retaining their separate data and measurement counts. The two minimum explanations have odds weight $xy^{\kappa-1}$. For even $\nu$, the next configurations comprise one $00$ explanation with one data fault and $3\nu/2$ measurement faults, and $N_A,N_B$ explanations in classes $10,01$, each with three data faults and $3\nu/2-2$ measurement faults. This gives Eq.~\eqref{eq:methods-ind-imbalance}.

For odd $\nu$, the weight-$(\nu+2)$ explanations in classes $00,11$ have odds weights $xy^{\nu+1}$ and $x^3y^{\nu-1}$. They need not cancel when $p\ne q$. At weight $2\nu+1$, a single data fault one measurement round beyond the endpoint pair gives a $00$ explanation with odds weight $xy^{2\nu}$. The boundary placement admits one such fault; the interior placement admits two, one on each side. The $10,01$ counts are unchanged by translation and give their stated difference. All omitted configurations have weight at least $2\nu+2$, so their total probability is $O(p^{2\nu+2})$ because $q=\Theta(p)$. The independent other sector contributes the same common factor to both placements.
\end{proof}

\begin{proof}[Proof of Theorem~\ref{thm:marginal-residual}]
Choose any minimum-degree qubit. For odd $\nu$, fix any central check $c$. Choose the remaining partition $\mathcal A,\mathcal B$ to minimize $|N_A-N_B|$ and orient it so that $\eta=N_A-N_B\geq0$.

If $\eta>0$, we show in Supplementary Information that it excludes both compatible outputs on $U$. The positive endpoint preference $\delta$ excludes $b$. For a check in $\mathcal A$ with the largest $n_h$, a negative measurement marginal bias (or a negative sum of two such biases for odd degree) excludes $a$. Posterior concentration rejects every bit outside $U$, so the boundary syndrome forces inconsistency.

If $\eta=0$ and $\nu$ is even, Eq.~\eqref{eq:methods-ind-imbalance} gives $\Delta=\Theta(p^{\kappa+\nu})>0$ and $\delta=O(p^{\kappa+\nu+1})$, again forcing inconsistency by Proposition~\ref{prop:methods-marginal-criterion}.

For $\eta=0$ and odd $\nu$, subtract the two expansions in Eq.~\eqref{eq:methods-two-rate-imbalance}:
\[
\Delta_1-\Delta_0=F_0xy^{2\nu}+O(p^{2\nu+2})
=\Theta(p^{2\nu+1}).
\]
Thus at least one placement satisfies $|\Delta_j|\geq|\Delta_1-\Delta_0|/2=\Omega(p^{2\nu+1})$, whereas both have $\delta_j=O(p^{2\nu+2})$. At each sufficiently small $p$, at least one of the two syndromes therefore forces equal MAP decisions on its data endpoints, violating their required odd parity. In every case, each candidate syndrome has probability
\[
\Pr(s_*^{(j)})=2F_0xy^{\kappa-1}+O(p^{\kappa+1})=\Theta(p^\kappa).
\]
Each case therefore supplies a syndrome of probability $\Theta(p^\kappa)$ on which MAP is inconsistent, proving $P_{\mathrm{inc}}(p)=\Omega(p^\kappa)$. In the balanced odd-degree case, the successful placement may depend on $p$, but the constants are uniform over the two choices. No limit or expansion for $q/p$ is needed.
\end{proof}

\begin{corollary}[Quartic inconsistency for the LP and Kasai code matrices]
\label{cor:methods-phenom-quartic}
For the LP code matrix with $n=1122$ and five noisy measurement rounds, and the Kasai code matrix with $n=1152$ and eight noisy measurement rounds, the independent $X/Z$ phenomenological model with data-fault probability $p$ and measurement-fault probability $q=\Theta(p)$ has $P_{\mathrm{inc}}(p)=\Omega(p^4)$ under exact marginal MAP.
\end{corollary}
\begin{proof}
The spatial checks provided in the Supplementary Information verify the hypotheses of Theorem~\ref{thm:marginal-residual}. Both matrices have a minimum column degree $\nu=3$. At the selected qubit, the other degree-3 incidence counts are $(4,4,4)$ for LP and $(11,11,11)$ for Kasai, so $N_A=N_B$. The measurement-round counts admit both placements. Direct enumeration in Supplementary Information confirms the separate data/measurement counts for each placement.
\end{proof}
All asymptotic statements hold at fixed code and measurement-round count. The theorem concerns exact marginals in the independent $X/Z$ model. Approximate NeuralMP decoding and depolarizing simulations are assessed separately by their numerical results.

\subsection{Diversification}

After iterative GNN decoding, almost all remaining failures leave a nonzero residual syndrome; syndrome-consistent but incorrect corrections are comparatively very rare. Similar cases arise in belief propagation and are commonly addressed by generating alternative decoding attempts and keeping the best consistent one~\cite{wang2025fullyparallelized,Koutsioumpas2025AutDEC,Baldelli2026QuasiDyadic}. We apply this principle to the trained NeuralMP: for shots whose residual syndrome remains nonzero, the network infers correction candidates with modified versions of the decoding problem and selects the one with the highest confidence.

\textbf{Syndrome-flip (SF) search.} We hypothesize a small set of error mechanisms, indicated by \(t\), and decode the modified syndrome \(s'=s\oplus Ht\) with the same iterative NeuralMP, obtaining \(e'\). The candidate for the original problem is \(e=t\oplus e'\), which satisfies \(He=s\) whenever \(He'=s'\). Because the network re-decodes the modified syndrome, it can revise the whole correction rather than only the hypothesized bits. The hypotheses are guided by the network's confidence through two signals: \emph{absz}, which ranks mechanisms by the smallest absolute logit of the first iteration\(|z_{0,f}|\) (locally ambiguous decisions), and \emph{flip}, which ranks mechanisms by how often they are corrected across the iterations (decisions that are revised repeatedly, the neural analog of the oscillation signal of BP-SF~\cite{wang2025fullyparallelized}). For each signal we take the 32 top-ranked mechanisms and form 32 single, 32 pair and 32 triple hypotheses, giving 192 candidates per shot.

\textbf{Further diversification.} For LP20 Code, if all SF candidates have non-sero residual syndrome, a second stage decodes the syndrome in equivalent representations given by code automorphism (\(HP_g=D_gH\), \(P_gp=p\)), for which the network is not exactly equivariant (because of positional encoding) and therefore follows different trajectories, combined with random multiplicative perturbations of the prior inputs, \(\tilde p_f=\mathrm{clip}(p_fe^{\alpha\xi_f})\), \(\xi_f\sim\mathcal N(0,1)\). After decoding, candidates are mapped back to the original coordinates.

\textbf{Selection.} Among candidates with \(He=s\) we select \(e^\star=\arg\min\sum_f e_f\log\frac{1-p_f}{p_f}\) under the original noise model; if none is consistent, the output of the fourth iterative decoding is retained. No ground-truth information is used.

The search converts the dominant failure mode, non-convergence, into a recoverable one: for the Kasai code at \(p=1.5\%\) it resolves 99.6\% of the shots left unresolved by iterative decoding, reducing the unresolved fraction from \(2.6\times10^{-4}\) to \(1.0\times10^{-6}\) per shot. Since these search methods have similar patterns, with the help of neutral network we can make candidates decoding as parallel batch, its latency can be amortized.

\subsection{Quantum memory benchmarks}
\label{sec:methods-benchmarks}

For LP20 code, we simulate $r=5$ measurement rounds with independent $X$ and $Z$ data errors, each occurring with a probability $p$, and measurement errors occurring with a probability $q=p$. For Kasai, we use $r=8$ measurement rounds with depolarizing data errors of probability $p$ and measurement errors of probability $q=4p/3$.

For LP20 code, we set the hidden dimension of NeuralMP as $\chi=128$ and 12 MP-rounds; for Kasai code, the network has a hidden dimension of $\chi=64$ and 80 MP-rounds. Both use at most four outer iterations. Shots with a nonzero residual enter syndrome-flip search, which tests 192 iterative-decoding candidates constructed using uncertainty and repeated-correction signals. For LP20 code, shots without a syndrome-consistent candidate enter an additional stage of 396 prior-perturbed, symmetry-transformed candidates and 96 syndrome-flip candidates guided by symmetry-view disagreement. And for the Kasai code, \texttt{Triage} uses no second search stage. Candidates are checked against the original syndrome and ranked using the original noise priors.

For LP20 code, BP-OSD baseline uses 100 min-sum iterations with scaling 0.625 and order-10 combination-sweep OSD. Its relay-BP baseline uses 100 initial BP iterations and up to 300 relay stages of 60 iterations each. For Kasai code, the two-stage baseline uses 100 min-sum iterations followed by up to 300 relay stages of 60 iterations; Tiered adds a COPT-based search, with a limit of 120 seconds per side and eight solver threads.

The per-qubit-per-round LER is $[1-(1-\bar P)^{1/r}]/k$, where $\bar P$ is the mean of the block failure probabilities for the two logical sectors and $k$ is the number of logical qubits. Error bars combine 95\% Clopper--Pearson intervals with uncertainty in the sector when the failed sector was not recorded; treating sector counts as independent in the interval calculation is an approximation.

For \texttt{Triage}, the plotted amortized time uses a batch size of 64 and combines measured front-end and search costs with observed search frequencies, including handoff and final processing. The neural decoder was measured on one NVIDIA H-series GPU; the algorithmic baselines were measured on the CPU, with single-threaded BP and relay-BP and eight threads for the Tiered solver.

\section*{Acknowledgements}
 We acknowledge the insightful discussion with Yi-Zhuang You. H.Y.H. acknowledges support from a startup fund from IOP, CAS. which initiated the study. C.Z. acknowledges support from a startup fund from Eastern Institute of Technology, Ningbo. We gratefully acknowledge the support of the Pioneer Program at the Shanghai Artificial Intelligence Laboratory and thank program managers Xiaoyi Yu and Ting Wang for their support and assistance. H.Y.H. also thanks Shang Liu for the GPU support. Part of the computational work for this research was supported by the High Performance Computing Center at Eastern Institute of Technology, Ningbo. We thank Cardinal Operations for providing an academic license for COPT, which we used to implement the most-likely-error (MLE) decoder. We acknowledge the use of AI in this project and writing.

\textbf{Hardware specifications}: all latency measurements were performed on a single server. NeuralMP’s decoding time was measured on a single NVIDIA H-series GPU, with GPU clocks left at their default dynamic boost. The host has two Intel Xeon Platinum 8481C CPUs (56 cores each, 2.0 GHz base, up to 3.8 GHz turbo, turbo enabled). The BP-based baselines (min-sum BP, relay-BP and BP-OSD) were run single-threaded, with each process pinned
to its own physical core. The exact-solver stage of the Tiered decoder (COPT) was run with 8 threads per call.

\textbf{Competing interests:} All authors declare no competing interests.

\textbf{Data availability:} Source data are available for this paper. All other data supporting the plots within this paper and other study
findings are available from the corresponding author upon reasonable request.

\textbf{Code availability: } The code used in this study is available from the corresponding author upon reasonable request.

\emph{Note: We will release the fully supplementary information in the second update soon.}

\let\oldaddcontentsline\addcontentsline
\renewcommand{\addcontentsline}[3]{}
\bibliography{main.bib}
\let\addcontentsline\oldaddcontentsline








\end{document}